%% file: main.tex
\documentclass[onecolumn,draftcls,12pt]{IEEEtran}

\input{setup/preamble.tex}
\input{setup/info.tex}
\begin{document}

\maketitle

\begin{abstract}
This work investigates an \emph{integrated sensing and edge artificial intelligence} (ISEA) system, where multiple devices first transmit probing signals for target sensing and then offload locally extracted features to the \emph{access point} (AP) via analog \emph{over-the-air computation} (AirComp) for collaborative inference. To characterize the relationship between AirComp error and inference performance, two proxies are established: the \emph{computation-optimal} proxy that minimizes the aggregation distortion, and the \emph{decision-optimal} proxy that maximizes the inter-class separability, respectively. Optimal transceiver designs in terms of closed-form power allocation are derived for both \emph{time-division multiplexing} (TDM) and \emph{frequency-division multiplexing} (FDM) settings, revealing threshold-based and dual-decomposition structures, respectively. Experimental results validate the theoretical findings.

\end{abstract}


\section{Introduction}

To support latency-sensitive applications such as smart factories, digital twins, and the low-altitude economy, the \emph{sixth generation} (6G) of wireless networks is expected to extend its functional scope by integrating sensing capabilities, thereby giving rise to the paradigm of \emph{integrated sensing and communication} (ISAC) \cite{liu2022integrated}. In this fashion, massive sensory data needs to be processed in real time, which naturally involves three tightly coupled processes, namely sensing, communication, and computation, as illustrated in Fig.~\ref{1Multi_User_Task_ISCC_intro}. From the computational perspective, intelligent computation is evolving toward a task-oriented paradigm, where the computation is directly aligned with specific tasks, such as classification and detection. This paradigm requires edge \emph{artificial intelligence} (AI) to enable collaborative computation, and is therefore referred to as \emph{integrated sensing and edge AI} (ISEA) \cite{liu2025integrated}.

An efficient ISEA system relies on two key aspects: system-level optimization and sensing-level enhancement \cite{liu2025integrated}. The former focuses on the joint optimization of sensing, communication, and computation processes, giving rise to a rich set of relevant techniques such as model partitioning \cite{yao2025energy}, time-slot and power allocation \cite{he2023integrated}, and transceiver cooperative design \cite{feres2023over}. The latter aims to improve data quality from the sensing source through joint optimization of sensing error\cite{wen2023task} and multi-view sensing aggregation \cite{chen2024view}. However, the performance improvement from both aspects is fundamentally constrained by the frequent transmission of high-dimensional features, which results in a severe communication bottleneck. 

\begin{figure}[!t]    
	\centering
	{\includegraphics[width=0.4\textwidth]{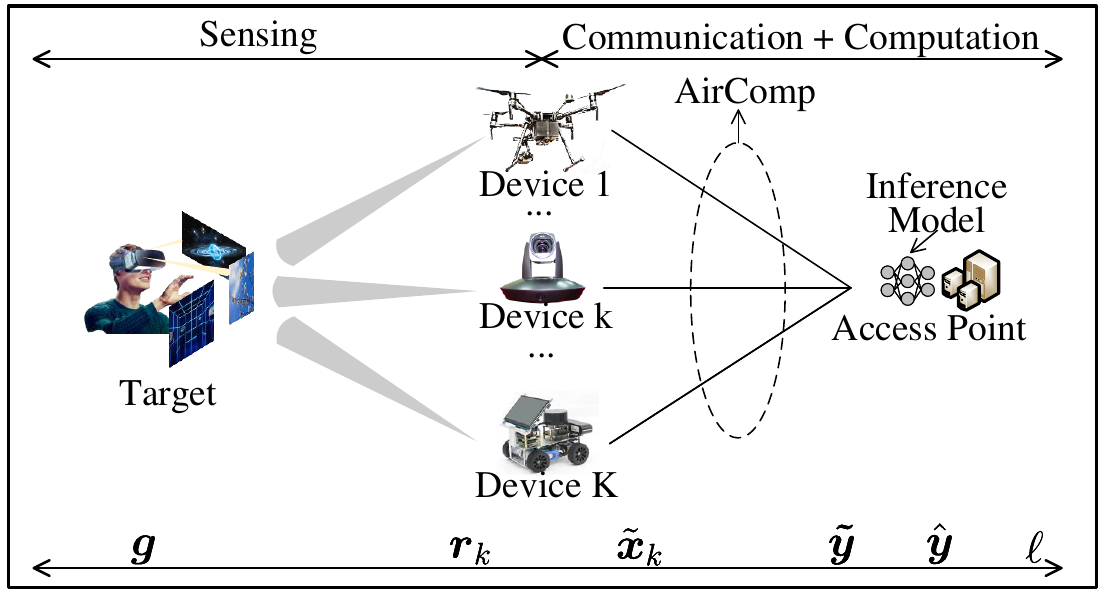}}
	\caption{The considered ISEA system consists of a single common target and one AP. Multiple devices first transmit probing signals for target sensing and then offload locally extracted features to the AP by AirComp.}
    \label{1Multi_User_Task_ISCC_intro}
    \vspace{-0.2cm}
\end{figure}

As a promising multiple-access technique, \emph{over-the-air computation} (AirComp) integrates multi-access transmission with nomographic functional computation (e.g., sum and averaging) by exploiting the signal superposition property of wireless multiple-access channels \cite{csahin2023survey,cao2020optimized,liu2020over}.
This enables AirComp to overcome the scalability bottleneck of conventional orthogonal multiple-access schemes, where radio resources are divided among users. Recent work has partially investigated AirComp for data aggregation in ISEA systems. Specifically, \cite{chen2024view} compared AirComp with orthogonal access schemes, \cite{liu2023over} investigated an AirComp-enabled max-value computation, and \cite{wen2023task} proposed a task-oriented AirComp framework for ISEA, where the \emph{Mahalanobis distance} (MD) between two classes is adopted as the discriminant metric. AirComp for ISEA is still in its infancy, and most existing studies focus on algorithmic design or proxy-based optimization of task performance \cite{chen2024view,csahin2023survey,liu2023over,wen2023task,cao2020optimized,liu2020over}. However, these schemes remain suboptimal, as they lack a unified framework to characterize the relationship between task performance and AirComp-induced errors. In particular, the objective of a task-oriented paradigm is not to reconstruct the raw sensory data but to preserve the task-relevant information embedded within it. Conventional computation-optimal AirComp designs \cite{chen2024view,csahin2023survey,liu2023over,cao2020optimized,liu2020over}, which overlook task characteristics, inevitably suffer performance degradation under adverse channel conditions. Meanwhile, even with the introduction of task-aware proxies (e.g., MD in \cite{wen2023task}), the performance gap between proxy-optimal and computation-optimal designs remains unclear.

Motivated by the aforementioned perspectives, this paper investigates a multi-user ISEA system with the objective of maximizing inference performance, as shown in Fig.~\ref{1Multi_User_Task_ISCC_intro}.
Here, inference refers to determining the category of a given sample and represents a task-oriented classification process.
The key contributions of this work are summarized as follows.
We first establish a theoretical relationship between the AirComp error and inference performance, providing a formal justification for two performance proxies, termed the computation-optimal and decision-optimal proxies. Optimal transceiver designs in terms of closed-form power allocation are derived for both \emph{time-division multiplexing} (TDM) and \emph{frequency-division multiplexing} (FDM) settings, revealing threshold-based and dual-decomposition structures. Extensive experimental results validate the theoretical findings.



\section{System Model} \label{sec:system_model}
Consider a multi-user ISEA system, as illustrated in Fig.~\ref{1Multi_User_Task_ISCC_intro}.
A set of $K$ distributed ISAC devices/users observe a common target for multi-view sensing. Each device locally extracts features from its observation and uploads them over a multi-access channel to the \emph{access point} (AP), which aggregates the features to enable collaborative inference. 
Key models and assumptions are detailed in the following subsections.

\subsection{Sensing Model}
\subsubsection{Sensing Process}Each device, say device $k$, employs \emph{frequency-modulated continuous-wave} (FMCW) modulation for target sensing, and the observation is denoted as \cite{wen2023task,yao2025energy}, 
\begin{equation}\label{Eq:receivedr_k}
    \boldsymbol{r}_k = \boldsymbol{g} + \boldsymbol{z}_k,
\end{equation}
where $\boldsymbol{g} \in \mathbb{C}^{F}$ represents the ground-truth sensory data of the target, and $\boldsymbol{z}_k \sim \mathcal{CN}(\boldsymbol{0}, \sigma_{k}^2 \mathbf{I}_{{F}})$ denotes the AWGN. Here, $F$ denotes the sensing dimension, i.e., the number of time–frequency sampling points in one FMCW frame.
\subsubsection{Feature Extraction}
The goal of the ISAC devices and the AP is to jointly estimate the class of the target. Due to the bandwidth limitation, each observation $\boldsymbol{r}_k$ must be compressed before transmission. Hence, we extract the low-dimensional feature from $\boldsymbol{r}_k$ using \emph{principal component analysis} (PCA) as \cite[Chapter 12]{bishop2006pattern} 

\begin{equation}\label{Eq:PCA}
\begin{aligned}
\tilde{\boldsymbol{x}}_k = \boldsymbol{U}^\top \boldsymbol{r}_k = \underbrace{\boldsymbol{U}^\top \boldsymbol{g}}_{\text{Ground-truth } \boldsymbol{x}} + \underbrace{\boldsymbol{U}^\top \boldsymbol{z}_k}_{\text{Sensing noise }\boldsymbol{d}_k} = \boldsymbol{x} + \boldsymbol{d}_k,
\end{aligned}
\end{equation}
where $\tilde{\boldsymbol{x}}_k \in \mathbb{C}^{M}$ is the feature vector obtained by a unitary matrix $\boldsymbol{U} \in \mathbb{C}^{F \times M}$, consisting of the ground-truth feature $\boldsymbol{x}$ and the sensing noise $\boldsymbol{d}_k$, and $M$ denotes the feature dimension with $M \le F$. Since the standard Gaussian distribution is unitary invariant, we have $\boldsymbol{d}_k \sim \mathcal{CN}(\bm{0}, \sigma_{k}^2 \mathbf{I}_{M})$. For tractability, the feature vector $\boldsymbol{x}$ of class $\ell \in \mathcal{L} = \{1, \cdots, {L}\}$, is modeled as a multivariate Gaussian distribution:
 \begin{equation}\label{Eq:class_istribution}
    p(\bm{x}|\ell) = \mathcal{CN}(\bm{x}\,|\,\boldsymbol{\mu}_\ell, \boldsymbol{\Sigma}),
\end{equation}
where $\boldsymbol{\mu}_\ell = [\mu_{\ell,1},\mu_{\ell,2},\cdots,\mu_{\ell,M}]^\top$ denotes mean vector and $\boldsymbol{\Sigma}= \mathrm{diag} (\sigma_1^2,\sigma_2^2,\cdots,\sigma_{M}^2)$ denotes diagonal covariance matrix $\boldsymbol{\Sigma} \in \mathbb{R}^{M \times M}$. Here, $L$ denotes the number of classes. The overall class distribution can thus be expressed as a linear superposition of Gaussian components:
$p(\bm{x}) = \frac{1}{L} \sum_{\ell=1}^{L} p(\bm{x}|\ell), \forall \ell\in\mathcal{L}.$ The distribution of the observed feature $\tilde{\boldsymbol{x}}_k$ can also be derived. Since the sensing noise is \emph{i.i.d.}, we have $    p(\tilde{\boldsymbol{x}}_k|\ell) \sim \mathcal{N}(\tilde{\boldsymbol{x}}_k\,|\,\boldsymbol{\mu}_\ell, \tilde{\boldsymbol{\Sigma}})$,
where $\tilde{\boldsymbol{\Sigma}} = \mathrm{diag} (\tilde{\sigma}_{k,1}^2,\tilde{\sigma}_{k,2}^2,\cdots,\tilde{\sigma}_{k,M}^2)$ with $\tilde{\sigma}_{k,m}^2=\tilde{\sigma}_{k,m}^2+\sigma_k^2.$ 

\subsection{Multi-Access Model}

Here, we adopt an analog AirComp transmission scheme for feature aggregation.
First, we present the key assumptions regarding the AirComp.
Based on these assumptions, we then introduce two feature transmission schemes and the corresponding receiver aggregation procedure.

\subsubsection{Key Assumptions}
Key assumptions about channel are summarized as follows. First, a frequency-selective slow fading channel with a coherence duration of $T_{\rm cd}$ is considered for all devices. Second, each device operates in \emph{time-division duplex} (TDD) mode and exploits channel reciprocity to acquire \emph{channel state information} (CSI) \cite[Chapter~5]{tse2005fundamentals}. The time interval between uplink and downlink transmissions satisfies $\Delta t_{\mathrm{TDD}} < T_{\mathrm{cd}}$, ensuring that the channel remains constant during the TDD cycle. In addition, AirComp assumes that all devices transmit their signals synchronously and exploit the signal superposition property of the wireless multiple-access channel to accomplish a specific computation task \cite{cao2020optimized}, e.g., 
\begin{equation}\label{Eq:ideal_received_signal}
    \boldsymbol{y} = \sum_{k=1}^{K} \tilde{\boldsymbol{x}}_k,
\end{equation}
which corresponds to a summation task, where $\boldsymbol{y}$ denotes the ideal received signal at the AP through joint transceiver design\footnote{For multi-user ISEA systems, common computation tasks include mean, maximum, and summation. Interested readers can refer to \cite{liu2025integrated} and references therein.}.

\subsubsection{Transmission Schemes}
To transmit the feature vector $\tilde{\boldsymbol{x}}_k = [\tilde{x}_{k,1}, \dots, \tilde{x}_{k,M}]^\top$ within one coherence block, two transmission schemes are considered: TDM and FDM \cite{carleial2003interference}.

In the TDM scheme, each feature element is sequentially transmitted over $M$ time slots. The received signal at the AP is  
\begin{equation}\label{eq:TDM_channel}
\tilde{y}_{t} = \sum_{k=1}^{K}h_{k,t}\, b_{k,t} \, \tilde{x}_{k,t} + w_t, \quad t=1,\dots,M,
\end{equation}
where $b_{k,t}$ denotes the transmit coefficient over slot $t$ at device $k$ and $w_t \sim \mathcal{CN}(0, \sigma_w^2)$ is the AWGN. We assume $M \leq T_{\mathrm{cd}}$ and thus the channel coefficient $h_{k,t}$ remains constant, i.e. $h_{k,t}=h_{k}$. 

In the FDM scheme, the feature elements are transmitted over $N$ subcarriers ($M\leq N$). The received signal is
\begin{equation}
\tilde{y}_{n} = \sum_{k=1}^{K}h_{k,n}\, b_{k,n} \, \tilde{x}_{k,n} + w_n, \quad n=1,\dots,M,
\end{equation}
where, with slight abuse of notation, $h_{k,n}$ denotes the channel response in frequency domain. Other notations follow the same definitions as the TDM scheme in \eqref{eq:TDM_channel}. 
The received feature vector is $
\boldsymbol{\tilde{y}} = [\tilde{y}_1, \tilde{y}_2, \ldots, \tilde{y}_{M}]^{\top},$
and the transmit coefficients satisfy the total power constraint
\begin{equation}\label{eq:power_constraint}
\sum_{n=1}^{M} |b_{k,n}|^2  \nu_{k,n}^2 \le P_{k},
\quad \nu_{k,n}^2 \triangleq \mathbb{E}[|\tilde{x}_{k,n}|^2],
\end{equation}
where $\nu_{k,n}^2$ can be estimated from offline training data samples \cite{wen2023task,chen2024view} and $P_{k}$ denotes the power budget for device $k$.

\subsubsection{Receiver Aggregation}
The AP performs linear aggregation with the receive coefficient $a_n$ as
\begin{equation}\label{Eq:aggregation_received_signal}
\hat{y}_n = a_n \tilde{y}_n
= a_n \sum_{k=1}^{K} h_{k,n} b_{k,n} \tilde{x}_{k,n} + a_n w_n.
\end{equation}
The aggregated feature vector is given by $\hat{\boldsymbol{y}} = [\hat{y}_1, \ldots, \hat{y}_{M}]^{\top}$, which is subsequently fed into the classification  model.

Finally, the error introduced by AirComp, i.e., the \emph{mean squared error} (MSE) between the ideal received signal defined in \eqref{Eq:ideal_received_signal} and the aggregated feature vector defined in \eqref{Eq:aggregation_received_signal}, is given by \cite{cao2020optimized}
\begin{equation}\label{Eq:AirComp_statistical_error}
\mathbb{E}[\|\boldsymbol{e}\|_2^2]
\!=\!\!
\sum_{n=1}^{M}\!\!\underbrace{\Big[ \!\sum_{k=1}^{K}\big|a_n h_{k,n} b_{k,n}-1\big|^2\tilde{\sigma}_{k,n}^2 \!\!+\!\!|a_n|^2\!\sigma_w^2\Big]}_{\mathrm{MSE}_n},
\end{equation}
where $\boldsymbol{e}\triangleq\hat{\boldsymbol{y}}-\boldsymbol{y}$ and  $\mathrm{MSE}_n$ denotes the MSE of the $n$-th subcarrier in FDM, or equivalently, the $t$-th time slot ($\mathrm{MSE}_t$) in TDM.

\subsection{Classification Model}
Consider a standard supervised classification model with a training set $\mathcal{S} = \{(\tilde{\boldsymbol{x}}_i, \ell_i)\}_{i=1}^{i=I},$
where $\tilde{\boldsymbol{x}}_i$ denotes the sensed features of the $i$-th observation defined in \eqref{Eq:PCA}, $\ell_i$ represents the corresponding class label defined in \eqref{Eq:class_istribution} and the total number of sensed samples is $I = K \times I_k$, where $K$ devices simultaneously perform $I_k$ sensing observations. The objective of training is to estimate the unknown function (classifier) $\ell=\mathrm{g}(\tilde{\boldsymbol{x}}_i )$ based on the given training set $\mathcal{S}$. The optimal classifier follows the \emph{maximum a posteriori} (MAP) rule as \cite{bishop2006pattern}
\begin{equation}
\mathrm{g}^\star(\tilde{\boldsymbol{x}}_i )=\arg \max _{\ell} p(\tilde{\boldsymbol{x}}_i |\ell).\label{eq:QptimalClassifier}
\end{equation}
Subsequently, the correct classification probability $\operatorname{Pr}(\ell\mid\mathrm{g}(\tilde{\boldsymbol{x}}_i))$ can be estimated for final decision making. The entire process from sensing to inference constitutes the following Markov chain:  ${\bm g}  \rightarrow {\bm r}_k \rightarrow
	\tilde{\boldsymbol{x}}_k\rightarrow
	\boldsymbol{\tilde{y}}\rightarrow \hat{\boldsymbol{y}}\rightarrow \ell$, as shown in Fig.~\ref{1Multi_User_Task_ISCC_intro}.
    
\section{Metrics}
To quantify the classification performance, we introduce two proxies, namely the computation-optimal and decision-optimal proxies. A theoretical justification for both proxies is provided by analyzing their relationship with the correct classification probability.

\subsection{Computation-Optimal Proxy}
The computation-optimal proxy is built upon the concept of the \emph{classification margin} $\gamma$ \cite{sokolic2017robust,liu2023over}, which establishes a bridge between the classification accuracy and the AirComp errors \eqref{Eq:AirComp_statistical_error} \footnote{The classification margin is an intrinsic property of a classifier, representing an inherent separation in the feature space.}. An illustrative example of classification margin is given in Fig.~\ref{2Fig2_geometric}.

\begin{theorem}[\cite{sokolic2017robust,liu2023over}]\label{theorem:computation_optimal}
    The classification accuracy at the AP is lower bounded as
    \begin{equation}
\!A_{\mathrm{AP}} \! \geq \!A_0 \operatorname{Pr}\!\!\left[\|\boldsymbol{e}\|<\gamma\right] 
\!\overset{(a)}{\geq} \!A_0\!\!\left(1-\frac{\mathbb{E}[\|\boldsymbol{e}\|^2]}{\gamma^2}\right),
    \end{equation}
    where $A_0$ denotes the inference performance of the classification model $g(\cdot)$ under noise-free transmission and (a) follows the Markov’s inequality. 
\end{theorem}
Theorem~\ref{theorem:computation_optimal} implies that a sufficient condition for correct classification is that the statistical error introduced by AirComp must be sufficiently small. Since $\gamma$ is fixed for a given classification model, minimizing the computation MSE in \eqref{Eq:AirComp_statistical_error} contributes to improved classification accuracy.

\begin{figure}[!h]
    \centering
    \includegraphics[width=0.28\textwidth]{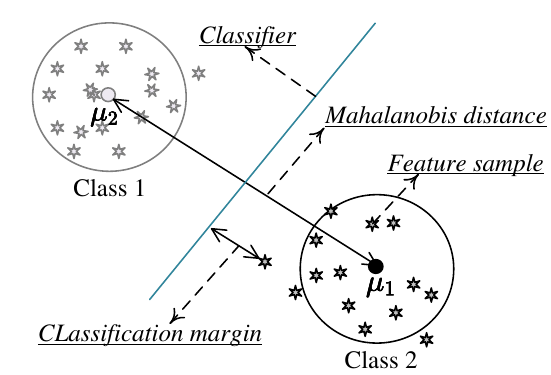}
    \caption{The geometric interpretation of classification margin and Mahalanobis distance.}
    \label{2Fig2_geometric}
\end{figure}
\subsection{Decision-Optimal Proxy}
The decision-optimal proxy is developed from the perspective of Bayesian hypothesis testing, which establishes a connection between MD and classification accuracy.
\begin{theorem}\label{theorem:Decision-Optimal_Proxy}
The MAP estimator achieves the minimum error probability in classifying $\bm{x}$, and the corresponding classification error probability asymptotically behaves as
\begin{equation}\label{eq:AsymptoticsPosterio2}
     \operatorname{P}_e \doteq\exp\left(-\kappa K \operatorname{G}_{\operatorname{min}}\right),
\end{equation}
where $\doteq$ denotes equivalence in exponential order, $\kappa$ is the constant coefficient and $\operatorname{G}_{\operatorname{min}}$ is the minimum inter-class MD between any two classes, i.e., $\operatorname{G}_{\operatorname{min}} = \underset{(\ell,\ell^{\prime}) \in \mathcal{L}}{\min}\operatorname{G}_{\ell,\ell^{\prime}}$.
\end{theorem}

\begin{proof}
    First, the classification error
probability $\operatorname{P}_e$ of any estimator is lower bounded by Fano’s inequality as \cite{kanaya1995asymptotics}
\begin{equation}\label{Eq:Fano}
H(\ell|{\boldsymbol{x}})
\leq H_b(\operatorname{P}_e) + \operatorname{P}_e \log (|\mathcal{L}|-1),
\end{equation}
where $\operatorname{P}_e \triangleq  1-\operatorname{Pr}(\ell\mid\mathrm{g}({\boldsymbol{x}}))$, $H(\ell|{\bm{x}})$ denotes the conditional entropy, $H_b(\cdot)$ is the binary entropy function and $|\mathcal{L}|$ is the cardinality of a set. For a sufficiently large observations, we have the following asymptotic relationship \cite{kanaya1995asymptotics}
\begin{equation}\label{eq:AsymptoticsPosterio}
    H(\ell|{\bm{x}}) \doteq P_e \doteq \exp\left(-K C\right),
\end{equation}
where $C=\underset{0\leq\rho\leq1}{\min} \mathsf{D}_{\mathrm{ch}}\left(p({\bm{x}}|\ell),p({\bm{x}} |\ell^{\prime})\right), \forall (\ell,\ell^{\prime}) \in \mathcal{L}$ denotes the minimum Chernoff information over all the possible class pairs, achieved by optimizing the parameter $\rho$ in the Chernoff distance $\mathsf{D}_{\mathrm{ch}}(\cdot,\cdot)$ \cite[P. 98]{fukunaga2013introduction}. Under the identical class covariances $\boldsymbol{\Sigma}$ (see \eqref{Eq:class_istribution}), the Bhattacharyya distance coincides with the optimal Chernoff distance when $\rho = 1/2$, and can be further simplified as 
\begin{equation}\label{eq:ch_Mahalanobis}
\!C \! = \kappa \operatorname{G}_{\operatorname{min}}, 
\end{equation}
where $\operatorname{G}_{\operatorname{min}}=\underset{(\ell,\ell^{\prime}) \in \mathcal{L}}{\min}\operatorname{G}_{\ell,\ell^{\prime}}({\bm{x}})$ is the minimum inter-class MD between any two classes and $\operatorname{G}_{\ell,\ell^{\prime}}(\bm{x})$ is
\begin{equation}\label{eq:Mahalanobis}
    \!\!\!\operatorname{G}_{\ell,\ell^{\prime}}(\bm{x})\!\!=\!\!\left( {\bm \mu}_{\ell}\! -\!{\bm \mu}_{\ell^{\prime}}\right)^\top \!{\bm \Sigma}^{-1}\!\! \left(  {\bm \mu}_{\ell}\! -\!{\bm \mu}_{\ell^{\prime}} \right)\!\!\overset{(a)}{=}\! \sum_{m=1}^{M}\operatorname{G}_{\ell,\ell^{\prime}}(x_m),\!\!\! 
\end{equation} where $\operatorname{G}_{\ell,\ell^{\prime}}(x_m) \triangleq\frac{(\mu_{\ell,m} - \mu_{\ell^{\prime},m})^2}{\sigma_m^2}$ and (a) follows from the statistical independence of the feature dimensions after PCA. By substituting \eqref{eq:ch_Mahalanobis} into \eqref{eq:AsymptoticsPosterio}, we have \eqref{eq:AsymptoticsPosterio2}.
\end{proof}

As implied by Theorem~\ref{theorem:Decision-Optimal_Proxy}, with sufficient independent observations in aggregation, the class uncertainty will decrease. Similarly, maximizing the inter-class MD also leads to reduced uncertainty and improved classification performance. An illustrative example of MD is
given in Fig.~\ref{2Fig2_geometric}.
\section{Transceiver Design under the TDM Setting}
\label{sec:Optimal_AirComp_TDM}
Building on the two aforementioned proxies, we derive the optimal transceiver designs for both computation-optimal and decision-optimal cases under the TDM setting.
\subsection{Computation-Optimal Case}
Under the computation-optimal case, the optimization problem can be formulated as
\begin{equation*}
(\mathrm{P}1) \min_{\{b_{k,t},a_t\}_{t=1}^{t=M}} \quad \sum_{t=1}^{M} \mathrm{MSE}_t, \quad
\text{s.t.}  \eqref{eq:power_constraint}.
\end{equation*}
Since the channel is slow fading (i.e., $h_{k,t}=h_{k}$, $b_{k,t}=b_{k}$, $\nu_{k,t}=\nu_{k}$, and $a_t=a$), $(\mathrm{P}1)$ reduces to independent per-slot optimizations as
\begin{equation*}
(\mathrm{P}2) \min_{b_{k},a} \quad  \mathrm{MSE}_t, \quad
\text{s.t.}  |b_{k}|^2 \nu_{k}^2 \le P_{k}.
\end{equation*}
$(\mathrm{P}2)$ has an optimal solution, as stated in the following theorem.
\addtolength{\topmargin}{0.051 in}
\begin{theorem}[\cite{cao2020optimized,liu2020over}]\label{theorem_MMSE_Optimal_TDM}
    Assuming $P_1|h_1| \le \cdots P_k|h_k| \le \cdots P_K|h_{K}|$, the optimal solution of $(\mathrm{P}2)$ is
    \begin{equation}\label{eq:threshold-based_structure}
        |b_k^\star|
=\begin{cases}
\dfrac{\sqrt P_k}{|\nu_k|}, & 1 \le k\le k^\star,\\[6pt]
\dfrac{1}{|a^\star|| h_k|}, & k^\star < k\le K,
\end{cases}
    \end{equation}
    where $a^\star=\frac{ \sqrt P_k\!\sum_{k=0}^{k^\star}h_k }
    { P_k \sum_{k=0}^{k^\star}h_{k}^2 + \sigma_w^2}$ and $k^\star$ denotes the threshold index that separates the full-power and channel-inversion transmission.
\end{theorem} 
\subsection{Decision-Optimal Case}
Unlike the computation-optimal case, the decision-optimal case aims to maximize the inter-class MD $\operatorname{G}_{\ell,\ell^{\prime}}(\boldsymbol{x})$, as defined in \eqref{eq:Mahalanobis}. The decision-optimal formulation can also be optimized independently across different time slots.
To this end, we first derive the expression of MD for each received feature element and then formulate the optimization problem.

Based on \eqref{Eq:aggregation_received_signal} and \eqref{eq:Mahalanobis}, the MD of the $m$-th received feature element can be expressed as
\begin{equation}\label{Eq:y_m_Mahalanobis}
    \operatorname{G}_{\operatorname{min}}({\hat y}_m)\!\! =\!\!\underset{(\ell,\ell^{\prime}) \in \mathcal{L}}{\min}\!\! \operatorname{G}_{\ell,\ell^{\prime}}({\hat y}_m)\!=\!\frac{\big|\sum_{k=1}^{k=K} h_k b_k\big|^2\Delta_m}{\sum_{k=1}^{k=K} \!\!|h_k b_k|^2\tilde{\sigma}_{k,m}^2\!+\!\sigma_w^2},
\end{equation}
where $\Delta_m \triangleq \underset{(\ell,\ell^{\prime}) \in \mathcal{L}}{\min}(\mu_{\ell,m} - \mu_{\ell^{\prime},m})^2$. 
It can be observed that the MD is independent of $a$. Hence, the optimization focuses on the $b_k$, as formulated below.
\begin{equation*}
(\mathrm{P}3) \max_{b_{k}} \quad  \eqref{Eq:y_m_Mahalanobis}, \quad
\text{s.t.}  |b_{k}|^2 \nu_{k}^2 \le P_{k}.
\end{equation*}
$(\mathrm{P}3)$ has an optimal solution, which is given in the following theorem.

\begin{theorem}
    $(\mathrm{P}3)$ has an optimal solution as 
    \begin{equation}\label{eq:p3_solution} 
        |b_k^\star|
    =\begin{cases}
    \dfrac{\sqrt P_k}{|\nu_k|}, & u_k\le \tau^\star,\\[6pt]
    \dfrac{\tau^\star}{|h_k|}, & u_k>\tau^\star,
    \end{cases}
    \end{equation}
    where $u_k \triangleq \dfrac{|h_k|\sqrt{P_k}}{\nu_k}$ and $\tau^\star$ denotes the capping threshold.   
\end{theorem}
\begin{proof}
    Similar to the computation-optimal case, we assume that $P_1|h_1| \le \cdots P_k|h_k| \le \cdots P_K|h_{K}|$. Define $\tilde\Delta=\tfrac{|\Delta_m|^2}{\tilde{\sigma}_{k,m}^2}$, $\sigma_{\sf{eq}}^2=\tfrac{\sigma_w^2}{\tilde{\sigma}_{k,m}^2}$, $c_k\triangleq |h_k||b_k|$ and $u_k \triangleq \frac{|h_k|\sqrt{P_k}}{\nu_k}$. Then, $(\mathrm{P}3)$ can be equivalently reformulated as
\begin{equation*}
(\mathrm{P}4) \!\!\quad
\max_{{c_k}}
F(\boldsymbol{c})
\!=\!\frac{\tilde\Delta(\sum_{k=1}^{k=K} c_k)^2}{\sum_{k=1}^{k=K} c_k^2+\sigma_{\sf{eq}}^2},\quad\!\!\!
\text{s.t.} 0\le c_k\le u_k,\ \forall k.
\end{equation*}
$(\mathrm{P}4)$ is a fractional programming problem \cite{boyd2004convex}. Define $\sum_{k=1}^{K} c_k = c_{\sf{sum}}$.
Since the denominator of $F(\boldsymbol{c})$ is a convex function, $(\mathrm{P}4)$ can be readily verified to be equivalent to 
\begin{equation*}
(\mathrm{P}5)\ 
\min_{\{c_k\}} \quad  \sum_{k=1}^{K} c_k^2, \qquad
\text{s.t.}  \sum_{k=1}^{K} c_k = c_{\sf{sum}}, \\[3pt]
 0 \le c_k \le u_k,\ \forall k.
\end{equation*}
$(\mathrm{P}5)$ is a quadratic programming problem \cite{boyd2004convex}, whose Lagrangian can be expressed as
\begin{equation*}
    \!\!\mathcal{L}_{\mathrm{P}5}\!=\!\!\sum_{k=1}^{K} c_k^2\!-\!\!\lambda\Big(\sum_{k=1}^{K} \!c_k-\!\!c_{\sf{sum}}\Big)
\!\!+\!\!\sum_{k=1}^{K}\!\alpha_k(c_k-u_k)\!-\!\sum_{k=1}^{K}\beta_k c_k,
\end{equation*}
where $\lambda$, $\alpha_k$, and $\beta_k$ denote the Lagrange multipliers.
According to the \emph{Karush–Kuhn–Tucker} (KKT) conditions \cite{boyd2004convex}, the optimal solution satisfies $    c_k^\star=\min\{u_k,\tau\}$,
which exhibits a capped-based structure, where weaker links (i.e., smaller $u_k$) saturate first. Without loss of generality, we assume that the first $j$ links are saturated, while the remaining $(K-j)$ links share a common capping threshold $\tau$, i.e.,
\begin{equation}\label{eq:capped-based_structure}
    c_k
=\begin{cases}
u_k, & 1 \le k\le j,\\[6pt]
\tau, & j < k\le K.
\end{cases}
\end{equation}
By substituting \eqref{eq:capped-based_structure} into $F(\boldsymbol{c})$, we have $
    F(\tau)
=\frac{\tilde\Delta\left(\sum_{k=1}^{j} u_k + (K-j)\tau\right)^2}{\left(\sum_{k=1}^{j} u_k^2\right) + (K-j)\tau^2 + \sigma_{\sf{eq}}^2}.$
Checking the first derivative of $F(\tau)$ yields the stationary point
\begin{equation}
    \tau_j=\frac{\sum_{k=1}^{j} u_k^2 + \sigma_{\sf{eq}}^2}{\sum_{k=1}^{j} u_k},j\ge1.
\end{equation}
If $\tau_j \in (u_j,u_{j+1})$, it serves as the unique maximizer of $F(\tau)$ within this segment. Otherwise, the segment-wise optimum occurs at the boundary, i.e., $\tau = u_j$ or $\tau = u_{j+1}$.
By comparing all segments, the global threshold $\tau^\star$ can be obtained.
Accordingly, the optimal $c_k^\star$ values are determined, and the corresponding transmit amplitudes $|b_k^\star|$ can be obtained as \eqref{eq:p3_solution}.
\end{proof}



\begin{remark}
By comparing \eqref{eq:p3_solution} and \eqref{eq:threshold-based_structure}, we observe that both proxies yield the same optimal solution structure. The underlying reason is that, in the TDM scheme, the temporal degrees of freedom do not provide additional gain, since all slots experience the same quasi-static channel realization. In this regard, irrespective of the chosen proxy, the resulting solution exhibits the same structural form.
\end{remark}

We notice that the discriminative importance is heterogeneous across feature dimensions. It is therefore desirable for the channel degrees of freedom to support corresponding multiplexing, enabling non-uniform resource allocation aligned with feature importance.
We next extend the analysis to the FDM case.
\section{Transceiver Design under the FDM Setting}\label{sec:Optimal_AirComp_TDM}
In this section, the transceiver designs under the FDM setting are derived for both computation-optimal and decision-optimal cases.
\subsection{Computation-Optimal Case}
Different from $(\mathrm{P}1)$, each feature element here is transmitted over different subcarriers rather than different time slots.
Due to the frequency-selective channel, a joint power constraint across all subcarriers should be considered. The corresponding optimization problem can be formulated as
\begin{equation*}
(\mathrm{P}6) \min_{\{b_{k,n},a_n\}_{n=1}^{n=M}} \quad \sum_{n=1}^{M} \mathrm{MSE}_n, \quad
\text{s.t.}  \eqref{eq:power_constraint}.
\end{equation*}
$(\mathrm{P}6)$ is non-convex due to the coupling between $b_{k,n}$ and $a_n$. Fortunately, it can be verified that $(\mathrm{P}6)$ satisfies the time-sharing condition \cite{yu2006dual}, and thus the duality gap becomes negligible when the number of subcarriers is large. Next, $(\mathrm{P}6)$ can be optimally solved using the Lagrange duality method. The Lagrangian and its dual function of $(\mathrm{P}6)$ are respectively given by 
\begin{equation}\label{eq:LagrangianP4}
\mathcal{L}_{\mathrm{P}6}
=\sum_{n=1}^{M}\left[
\Phi_n(a_n,{b_{k,n}};\boldsymbol\lambda)
\right]
-\sum_{k=1}^{K}\lambda_k P_{k},
\end{equation}
\begin{equation}\label{eq:dualP4}
    \!\!g(\boldsymbol\lambda)\!\!=\!\!\!\!\!\inf_{{a_n},{b_{k,n}}}\!\!\mathcal{L}_{\mathrm{P}6}\!
=\!\!\!\sum_{n=1}^{M}\underbrace{
\!\inf_{a_n,{b_{k,n}}}\Phi_n(a_n,{b_{k,n}};\boldsymbol\lambda)
}_{\psi_n(\boldsymbol\lambda)}
\!-\!\!\!\sum_{k=0}^{K-1}\lambda_k P_k,
\end{equation}
where $\Phi_n(a_n,{b_{k,n}};\boldsymbol\lambda)\triangleq\sum_{k=1}^{K}|1-a_n h_{k,n} b_{k,n}|^2\tilde{\sigma}_{k,n}^2
+|a_n|^2\sigma_w^2
+\sum_{k=1}^{K}\lambda_k \nu_{k,n}^2|b_{k,n}|^2$ and $\boldsymbol{\lambda} = [\lambda_1,\ldots,\lambda_K]^{T} \succeq \mathbf{0}$ collects the dual variables corresponding to the individual power constraints of the users.
For a fixed $a_n$, minimizing $\Phi_n$ w.r.t. $b_{k,n}$ yields
\begin{equation}\label{eq:minimizing_b_k_n}
    |b_{k,n}^\star(\lambda_k,a_n)|
=\frac{\tilde{\sigma}_{k,n}^2 |a_n h_{k,n}|}
{\tilde{\sigma}_{k,n}^2|a_n|^2|h_{k,n}|^2+\lambda_k\nu_{k,n}^2}.
\end{equation}
Substituting \eqref{eq:minimizing_b_k_n} back $\Phi_n(a_n,{b_{k,n}};\boldsymbol\lambda)$ gives
\begin{equation*}
    \!\!\!\tilde{\sigma}_{k,n}^2|\!1\!\!-\!a_n h_{k,n} b_{k,n}^\star|^2\!\!\!+\!\!\lambda_k\nu_{k,n}^2|b_{k,n}^\star|^2
\!\!=\!\!\frac{\lambda_k\nu_{k,n}^2\tilde{\sigma}_{k,n}^2}
{\!\tilde{\sigma}_{k,n}^2|a_n|^2|h_{k,n}|^2\!\!\!+\!\!\!\lambda_k\nu_{k,n}^2\!}.
\end{equation*}
Hence, the dual term for each subcarrier is
\begin{equation*}
    \!\psi_n(\boldsymbol\lambda)\!
=\!\min_{a_n}\!\!
\underbrace{
\sum_{k=0}^{K-1}\!
\frac{\lambda_k\nu_{k,n}^2\tilde{\sigma}_{k,n}^2}
{\tilde{\sigma}_{k,n}^2|a_n|^2|h_{k,n}|^2+\lambda_k\nu_{k,n}^2}
\!+\!\!|a_n|^2\sigma_w^2
}_{\phi_n(|a_n|^2;\boldsymbol\lambda)}.
\end{equation*}
Let $r_n=|a_n|^2\ge0$,  minimizing $\phi_n(r_n;\boldsymbol\lambda)$ w.r.t. $r_n$ yields
\begin{equation}\label{eq:scalar_equation}
    \sum_{k=0}^{K-1}
\frac{\lambda_k\nu_{k,n}^2\sigma_n^4|h_{k,n}|^2}
{(\tilde{\sigma}_{k,n}^2 r_n|h_{k,n}|^2+\lambda_k\nu_{k,n}^2)^2}
=\sigma_w^2.
\end{equation}
which has a unique positive root $r_n^\star(\boldsymbol\lambda)$. Substituting it into \eqref{eq:dualP4} gives
$ g(\boldsymbol\lambda)
=\sum_{n=0}^{N-1}
\phi_n\big(r_n^\star(\boldsymbol\lambda);\boldsymbol\lambda\big)
-\sum_{k=0}^{K-1}\lambda_k P_{\mathsf c}.$
Base on \eqref{eq:scalar_equation}, \eqref{eq:minimizing_b_k_n} and \eqref{eq:power_constraint}, the dual problem can be solved by the subgradient-based methods such as the ellipsoid method \cite{boyd2004convex}.

\subsection{Decision-Optimal Case}
For the decision-optimal case, the objective is to maximize the MD of the received features across all subcarriers. The corresponding optimization problem can be formulated as
\begin{equation*}
(\mathrm{P}7) \min_{\{b_{k,n},a_n\}_{n=1}^{n=M}} \quad \sum_{n=1}^{M} \operatorname{G}_{\operatorname{min}}({\hat y}_n), \quad
\text{s.t.}  \eqref{eq:power_constraint}.
\end{equation*}

$(\mathrm{P}7)$ is also non-convex due to the coupling between ${b_{k,n}}$ across users and subcarriers. Similar to $(\mathrm{P}6)$, $(\mathrm{P}7)$ satisfies the time-sharing condition, and can thus be solved using the Lagrange duality method.
Following the same procedure as in \eqref{eq:minimizing_b_k_n} and \eqref{eq:scalar_equation}, we have
\begin{equation}\label{eq:minimizing_b_k_n_dg}
    |b_{k,n}^\star|
=\frac{|\Delta_n|^2 h_{k,n} z_n}
{\lambda_k\nu_{k,n}^2+|\Delta_n|^2\tilde{\sigma}_{k,n}^2 h_{k,n}^2 z_n^2},
\end{equation}
\begin{equation}\label{eq:consistency_condition}
    \sum_{k=1}^{K}
\frac{\lambda_k h_{k,n}^2\nu_{k,n}^2}
{(\lambda_k\nu_{k,n}^2+|\Delta_n|^2\tilde{\sigma}_{k,n}^2 h_{k,n}^2 z_n^2)^2}
=\frac{\sigma_w^2}{|\Delta_n|^2},
\end{equation}
where $z_n\triangleq \frac{\sum_{k=1}^{K}|h_{k,n} b_{k,n}|}{\tilde{\sigma}_{k,n}^2\sum_{k=1}^{K}|h_{k,n} b_{k,n}|^2+\sigma_w^2}$. Then, the outer-layer variables ${\lambda_k}$ are updated to satisfy the power constraints in \eqref{eq:power_constraint} based on \eqref{eq:minimizing_b_k_n_dg}, while the inner layer solves for the optimal $z_n^\star$ on each subcarrier according to \eqref{eq:consistency_condition}.

\begin{figure*}[!th]
	\centering
        \begin{minipage}{0.32\textwidth}
		{\includegraphics[width=\textwidth]
         {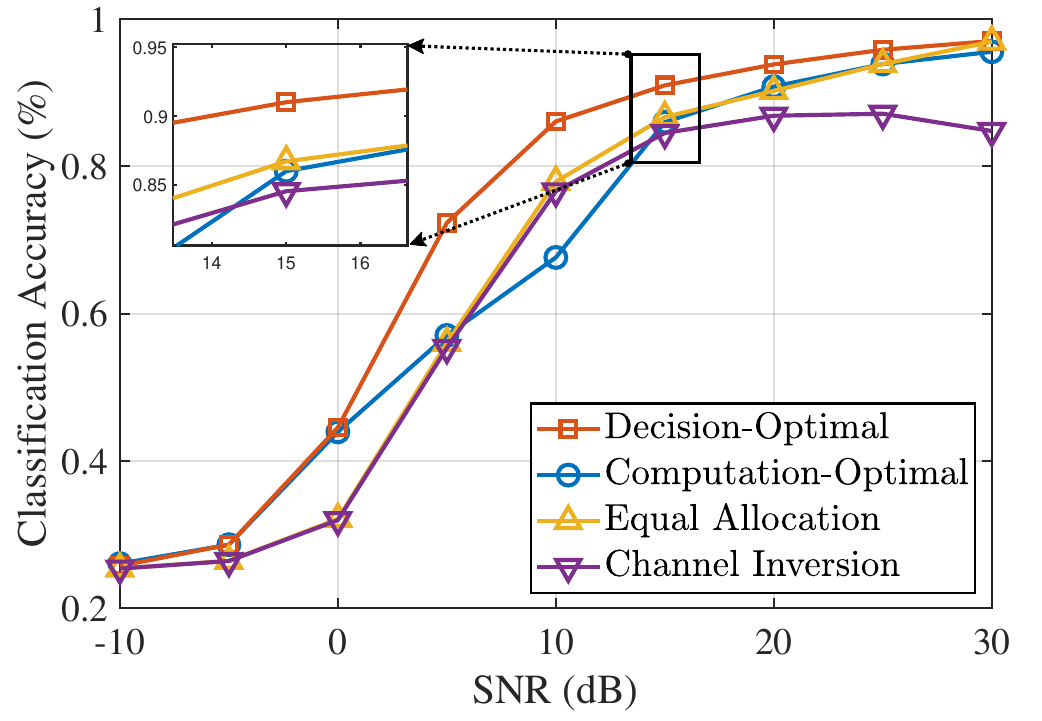}}
		\caption*{(a)}\label{Fig_P1_acc_snr}
	\end{minipage}
            \begin{minipage}{0.32\textwidth}
		{\includegraphics[width=\textwidth]
         {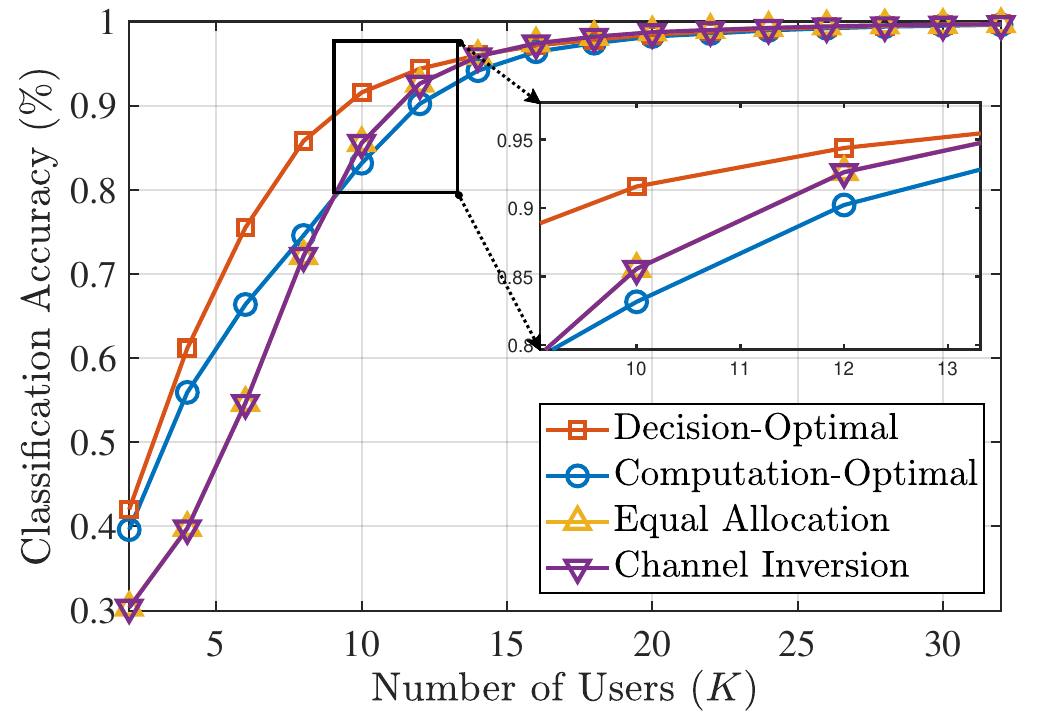}}
		\caption*{(b)}\label{Fig_P1_acc_k}
	\end{minipage}
                \begin{minipage}{0.32\textwidth}
		{\includegraphics[width=\textwidth]
         {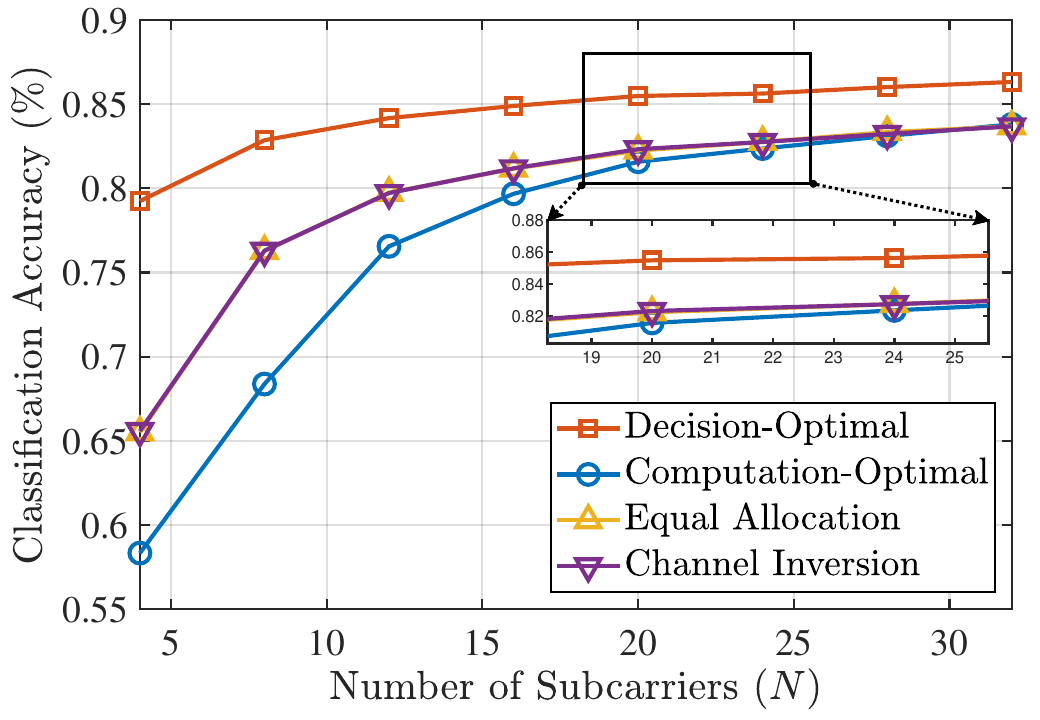}}
		\caption*{(c)}\label{Fig_P1_acc_n}
	\end{minipage}
\caption{Inference performance comparison under different system parameters: (a) varying $\mathrm{SNR}$, (b) varying number of users $K$, and (c) varying number of subcarriers $N$.}\label{Fig_results}
\vspace{-0.25cm}
\end{figure*}

\section{Numerical Results} \label{sec:results}

In this section, numerical results are presented to validate the analytical results. Besides the computation-optimal and decision-optimal cases, two baselines are considered: (1) \emph{Equal Allocation}, with uniform power allocation across all subcarriers, and (2) \emph{Channel Inversion}, with $|b_{k,n}| =\{\sqrt{\frac{P_k}{N\nu_{k,n}^2}},\ \frac{1}{|h_{k,n}|}\}$. We adopt a \emph{support vector machine} (SVM) as the edge AI model, which is trained on noise-free PCA features extracted from a human posture recognition dataset \cite{wen2023task}. 
The task is formulated as a four-class human posture classification problem. Key system parameters are configured as follows: $\sigma_w^2 = 0.1$, $K = 3$, $P \in [-10, 30]$ dB, $M = 4$, and $N = 32$.


In Fig.~\ref{Fig_results}(a), the inference performance is evaluated across a range of \emph{signal-to-noise ratio} (SNR) levels. Three key observations can be made. First, at extremely low SNR levels, the classification performance of all schemes degrades to random guessing, i.e., $\tfrac{1}{|\mathcal{L}|}$. Second, at high SNR levels, all schemes except \emph{Channel Inversion} approach the saturation upper bound. This is because \emph{Channel Inversion} attempts to compensate all subcarriers, but the total power constraint across all subcarriers limits its effectiveness. Third, in the moderate-SNR regime (i.e., $[0,15]$ dB), the computation-optimal scheme yields superior inference accuracy, benefiting from discriminant-aware and non-uniform resource allocation across feature dimensions. In Fig.~\ref{Fig_results}(b), the inference performance is evaluated under varying numbers of users $K$. The results validate Theorem~\ref{theorem:Decision-Optimal_Proxy}, showing that with sufficient independent observations in aggregation, the class uncertainty decreases accordingly. Moreover, as $K$ increases, all schemes converge to a common performance upper bound. This is because all schemes optimize each user independently, and no power coupling exists among users. In Fig.~\ref{Fig_results}(c), the inference performance under different numbers of subcarriers $N$ is presented. The performance gap between the decision-optimal scheme and the other schemes gradually narrows as $N$ increases. This is because the frequency-domain heterogeneity of the channel diminishes when $N$ becomes large, causing all schemes to approach the same time-sharing bound. However, a constant performance gap remains, since the decision-optimal design allocates more power to discriminative subcarriers with larger $|\Delta_n|^2 / \tilde{\sigma}_{k,n}^2$, whereas the computation-optimal design follows variance-based regularization. 

\section{Conclusion} \label{sec:conclusion}
This paper have investigated a multi-user ISEA system enabled by AirComp. Two theoretical proxies were introduced: the computation-optimal and decision-optimal proxies. Optimal transceiver designs in terms of closed-form power allocation were derived for both TDM and FDM settings, revealing threshold-based and dual-decomposition structures that jointly balance feature importance, channel conditions, and power budgets. Results show that the computation-optimal and decision-optimal designs coincide under quasi-static channels but diverge under frequency-selective conditions, where discriminant-aware allocation yields additional inference gains.


\bibliographystyle{IEEEtran}
\bibliography{reference/mybib}

\end{document}

%% file: setup/preamble.tex
\usepackage{amssymb}
\usepackage{amsmath}
\usepackage{cite}
\usepackage{url}
\usepackage{xcolor}
\usepackage{cite,graphicx,amsmath,amssymb}
\usepackage{fancyhdr}
\usepackage{mdwmath}
\usepackage{mdwtab}
\usepackage{caption}
\usepackage{amsthm}
\usepackage{bm}
\usepackage{amsmath}   
\usepackage{amssymb}   
\usepackage{mathrsfs}
\usepackage{setspace}
\usepackage{algorithm}
\usepackage{algorithmic}
\usepackage{pdfpages}
\usepackage{mathtools}
\usepackage{subcaption}
\usepackage{bbm}
\usepackage{comment}

\graphicspath{{./src/img/}}

\newtheorem{remark}{Remark}
\newtheorem{theorem}{Theorem}

\newtheorem{lemma}{Lemma}

\newtheorem{corollary}{Corollary}

\allowdisplaybreaks
\expandafter\def\expandafter\normalsize\expandafter{%
    \normalsize%
    \setlength\abovedisplayskip{4pt}%
    \setlength\belowdisplayskip{4pt}%
    \setlength\abovedisplayshortskip{2pt}%
    \setlength\belowdisplayshortskip{2pt}%
}

%% file: setup/info.tex
\title{Optimized Power Control for Multi-User Integrated Sensing and Edge AI}
\author{
    \IEEEauthorblockN{Biao Dong and Bin Cao\\
\IEEEauthorblockA{Harbin Institute of Technology, Shenzhen, China
\\E-mail: 23b952012@stu.hit.edu.cn and caobin@hit.edu.cn}}
\vspace{-0.7cm}
}

%% file: main.bbl
\begin{thebibliography}{10}
\providecommand{\url}[1]{#1}
\csname url@samestyle\endcsname
\providecommand{\newblock}{\relax}
\providecommand{\bibinfo}[2]{#2}
\providecommand{\BIBentrySTDinterwordspacing}{\spaceskip=0pt\relax}
\providecommand{\BIBentryALTinterwordstretchfactor}{4}
\providecommand{\BIBentryALTinterwordspacing}{\spaceskip=\fontdimen2\font plus
\BIBentryALTinterwordstretchfactor\fontdimen3\font minus \fontdimen4\font\relax}
\providecommand{\BIBforeignlanguage}[2]{{%
\expandafter\ifx\csname l@#1\endcsname\relax
\typeout{** WARNING: IEEEtran.bst: No hyphenation pattern has been}%
\typeout{** loaded for the language `#1'. Using the pattern for}%
\typeout{** the default language instead.}%
\else
\language=\csname l@#1\endcsname
\fi
#2}}
\providecommand{\BIBdecl}{\relax}
\BIBdecl

\bibitem{liu2022integrated}
F.~Liu, Y.~Cui, C.~Masouros, J.~Xu, T.~X. Han, Y.~C. Eldar, and S.~Buzzi, ``Integrated sensing and communications: Toward dual-functional wireless networks for 6g and beyond,'' \emph{{IEEE} J. Sel. Areas Commun.}, vol.~40, no.~6, pp. 1728--1767, 2022.

\bibitem{liu2025integrated}
Z.~Liu, X.~Chen, H.~Wu, Z.~Wang, X.~Chen, D.~Niyato, and K.~Huang, ``Integrated sensing and edge {AI}: {Realizing} intelligent perception in {6G},'' \emph{IEEE Comm. Surv. Tutor.}, vol.~28, pp. 2725--2770, 2026.

\bibitem{yao2025energy}
J.~Yao, W.~Xu, G.~Zhu, K.~Huang, and S.~Cui, ``Energy-efficient edge inference in integrated sensing, communication, and computation networks,'' \emph{{IEEE} J. Sel. Areas Commun.}, early access, May 2025, doi: 10.1109/JSAC.2025.3574612.

\bibitem{he2023integrated}
Y.~He, G.~Yu, Y.~Cai, and H.~Luo, ``Integrated sensing, computation, and communication: System framework and performance optimization,'' \emph{IEEE Trans. Wireless Commun.}, vol.~23, no.~2, pp. 1114--1128, 2023.

\bibitem{feres2023over}
C.~Feres, B.~C. Levy, and Z.~Ding, ``Over-the-air multisensor collaboration for resource efficient joint detection,'' \emph{IEEE Trans. Signal Process.}, vol.~72, pp. 384--399, 2023.

\bibitem{wen2023task}
D.~Wen, X.~Jiao, P.~Liu, G.~Zhu, Y.~Shi, and K.~Huang, ``Task-oriented over-the-air computation for multi-device edge {AI},'' \emph{IEEE Trans. Wireless Commun.}, vol.~23, no.~3, pp. 2039--2053, 2023.

\bibitem{chen2024view}
X.~Chen, K.~B. Letaief, and K.~Huang, ``On the view-and-channel aggregation gain in integrated sensing and edge {AI},'' \emph{IEEE J. Sel. Areas Commun.}, vol.~42, no.~9, pp. 2292--2305, 2024.

\bibitem{csahin2023survey}
A.~{\c{S}}ahin and R.~Yang, ``A survey on over-the-air computation,'' \emph{IEEE Comm. Surv. Tutor.}, vol.~25, no.~3, pp. 1877--1908, 2023.

\bibitem{cao2020optimized}
X.~Cao, G.~Zhu, J.~Xu, and K.~Huang, ``Optimized power control for over-the-air computation in fading channels,'' \emph{IEEE Trans. Wireless Commun.}, vol.~19, no.~11, pp. 7498--7513, 2020.

\bibitem{liu2020over}
W.~Liu, X.~Zang, Y.~Li, and B.~Vucetic, ``Over-the-air computation systems: Optimization, analysis and scaling laws,'' \emph{IEEE Trans. Wireless Commun.}, vol.~19, no.~8, pp. 5488--5502, 2020.

\bibitem{liu2023over}
Z.~Liu, Q.~Lan, A.~E. Kal{\o}r, P.~Popovski, and K.~Huang, ``Over-the-air multi-view pooling for distributed sensing,'' \emph{IEEE Trans. Wireless Commun.}, vol.~23, no.~7, pp. 7652--7667, 2023.

\bibitem{bishop2006pattern}
C.~M. Bishop and N.~M. Nasrabadi, \emph{Pattern recognition and machine learning}.\hskip 1em plus 0.5em minus 0.4em\relax New York, USA: Springer Science \& Business Media, 2006.

\bibitem{tse2005fundamentals}
D.~Tse and P.~Viswanath, \emph{Fundamentals of wireless communication}.\hskip 1em plus 0.5em minus 0.4em\relax Cambridge, U.K.: Cambridge Univ. Press, 2005.

\bibitem{carleial2003interference}
A.~Carleial, ``Interference channels,'' \emph{IEEE Trans. Inf. Theory}, vol.~24, no.~1, pp. 60--70, 2003.

\bibitem{sokolic2017robust}
J.~Sokoli{\'c}, R.~Giryes, G.~Sapiro, and M.~R. Rodrigues, ``Robust large margin deep neural networks,'' \emph{IEEE Trans. Signal Process.}, vol.~65, no.~16, pp. 4265--4280, 2017.

\bibitem{kanaya1995asymptotics}
F.~Kanaya \emph{et~al.}, ``The asymptotics of posterior entropy and error probability for {B}ayesian estimation,'' \emph{IEEE Trans. Inf. Theory}, vol.~41, no.~6, pp. 1988--1992, 1995.

\bibitem{fukunaga2013introduction}
K.~Fukunaga, \emph{Introduction to statistical pattern recognition}.\hskip 1em plus 0.5em minus 0.4em\relax Amsterdam, The Netherlands: Elsevier, 2013.

\bibitem{boyd2004convex}
S.~P. Boyd and L.~Vandenberghe, \emph{Convex optimization}.\hskip 1em plus 0.5em minus 0.4em\relax Cambridge, U.K.: Cambridge Univ. Press, 2004.

\bibitem{yu2006dual}
W.~Yu and R.~Lui, ``Dual methods for nonconvex spectrum optimization of multicarrier systems,'' \emph{IEEE Trans. Wireless Commun.}, vol.~54, no.~7, pp. 1310--1322, 2006.

\end{thebibliography}
